\documentclass[preprint,superscriptaddress,
showpacs,preprintnumbers,amsmath,amssymb]{revtex4}

\usepackage{amssymb,amsmath,amsthm}
\usepackage{graphicx}
\newtheorem{Thm}{Theorem}

\newtheorem{Prop}{Proposition}
\theoremstyle{definition}

\newcommand{\bra}[1]{{\left\langle #1 \right|}}
\newcommand{\ket}[1]{{\left| #1 \right\rangle}}

\newcommand{\T}{\mbox{$\mathrm{tr}$}}

\begin{document}
\title{Entanglement of formation and monogamy of multi-party quantum entanglement}
\author{Jeong San Kim}
\email{freddie1@khu.ac.kr} \affiliation{
 Department of Applied Mathematics and Institute of Natural Sciences, Kyung Hee University, Yongin-si, Gyeonggi-do 446-701, Korea
}
\date{\today}

\begin{abstract}
We provide a sufficient condition for the monogamy inequality of multi-party quantum entanglement of arbitrary dimensions in terms of entanglement of formation. Based on the classical-classical-quantum(ccq) states whose quantum parts are obtained from the two-party reduced density matrices of a three-party quantum state, we show the additivity of the mutual information of the ccq states guarantees the monogamy inequality of the three-party pure state in terms of EoF. After illustrating the result with some examples, we generalize our result of three-party systems into any multi-party systems of arbitrary dimensions.
\end{abstract}

\maketitle

\section{Introduction}
Quantum entanglement is a non-classical nature of quantum mechanics, which is a useful resource in many quantum information processing tasks such as quantum teleportation, dense coding and quantum cryptography~\cite{tele,qkd,rev}. Because of its important roles in the field of quantum information and computation theory, there has been a significant amount of research focused on quantification of entanglement in bipartite quantum systems. {\em Entanglement of formation}(EoF) is the most well-known bipartite entanglement measure with an operational meaning that asymptotically quantifies how many bell states are needed to prepare the given state using local quantum operations and classical communications~\cite{bdsw}. Although EoF is defined in any bipartite quantum systems of arbitrary dimension, its definition for mixed states is based on `convex-roof extension', which takes the minimum average over all pure-state decompositions of the given state. As such an optimization is hard to deal with, analytic evaluation of EoF is known only in two-qubit systems~\cite{wootters} and some restricted cases of higher-dimensional systems so far.

In multi-party quantum systems, entanglement shows a distinct behavior that does not have any classical counterpart; if a pair of parties in a multi-party quantum system is maximally entangled, then they cannot be entanglement, not even classically correlated, with the rest parties. This restriction of sharing entanglement in multi-party quantum systems is known as the {\em monogamy of entanglement}(MoE)~\cite{T04, KGS}. MoE plays an important role such as the security proof of quantum key distribution in quantum cryptography~\cite{qkd, qkd2} and the $N$-representability problem for fermions in condensed-matter physics~\cite{anti}.

Mathematically, MoE can be characterized using {\em monogamy inequality}; for a three-party quantum state $\rho_{ABC}$ and its two-party reduced density matrices $\rho_{AB}$ and $\rho_{AC}$,
\begin{align}
E\left(\rho_{A(BC)}\right)\geq E\left(\rho_{AB}\right)+E\left(\rho_{AC}\right)
\label{MoEg}
\end{align}
where $E\left(\rho_{XY}\right)$ is an entanglement measure quantifying the amount of entanglement between subsystems $X$ and $Y$ of the bipartite quantum state $\rho_{XY}$. Inequality~(\ref{MoEg}) shows the mutually exclusive nature of bipartite entanglement $E\left(\rho_{AB}\right)$ and $E\left(\rho_{AC}\right)$ shared in three-party quantum systems so that their summation cannot exceeds the total entanglement $E\left(\rho_{A(BC)}\right)$.

Using tangle~\cite{CKW} as the bipartite entanglement measure, Inequality~(\ref{MoEg}) was first shown to be true for all three-qubit states, and generalized for multi-qubit systems as well as some cases of higher-dimensional quantum systems~\cite{OV, KDS}. However, not all bipartite entanglement measures can characterize MoE in forms of Inequality (\ref{MoEg}), but only few measures are known so far satisfying such monogamy inequality~\cite{KSRenyi, KimT, KSU}. Although EoF is the most natural bipartite entanglement measure with the operational meaning in quantum state preparation, EoF is known to fail in characterizing MoE as the monogamy inequality in~(\ref{MoEg}) even in three-qubit systems; there exists quantum states in three-qubit systems violating Inequality~(\ref{MoEg}) if EoF is used as the bipartite entanglement measure. Thus, a natural question we can ask is `On what condition does the monogamy inequality hold in terms of the given bipartite entanglement measure?'

Here, we provide a sufficient condition that monogamy inequality of quantum entanglement holds in terms of EoF in multi-party, arbitrary dimensional quantum systems. For a three-party quantum state, we first consider the classical-classical-quantum(ccq) states whose quantum parts are obtained from the two-party reduced density matrices of the three-party state. By evaluating quantum mutual information of the ccq states as well as their reduced density matrices, we show that the additivity of the mutual information of the ccq states guarantees the monogamy inequality of the three-party quantum state in terms of EoF. We provide some examples of three-party pure state to illustrate our result, and we generalize our result of three-party systems into any multi-party systems of arbitrary dimensions.

This paper is organized as follows. First we briefly review the definitions of classical and quantum correlations in bipartite quantum systems and recall their trade-off relation in three-party quantum systems. After providing the definition of ccq states as well as  their mutual information between classical and quantum parts, we establish the monogamy inequality of three-party quantum entanglement in arbitrary dimensional quantum systems in terms of EoF conditioned on the additivity of the mutual information for the ccq states. We also illustrate our result of monogamy inequality in three-party quantum systems with some examples, and we generalize our result of entanglement monogamy inequality into multi-party quantum systems of arbitrary dimensions. Finally, we summarize our results.

\section{Correlations in bipartite quantum systems}
\label{Sec: correlations}
For a bipartite pure state $\ket{\psi}_{AB}$, its {\em entanglement of formation}(EoF) is defined by the entropy of a subsystem,
\begin{equation}
{E}_{\bf f}\left(\ket{\psi}_{AB} \right)=S(\rho_A),
\label{EoFpure}
\end{equation}
where $\rho_A=\T _{B} \ket{\psi}_{AB}\bra{\psi}$ is the reduced density matrix of $\ket{\psi}_{AB}$ on subsystem $A$,
and $S\left(\rho\right)=-\T \rho \ln \rho$ is the von Neumann entropy of the quantum state $\rho$.
For a bipartite mixed state $\rho_{AB}$, its EoF is defined by the minimum average entanglement
\begin{equation}
E_{\bf f}\left(\rho_{AB} \right)=\min \sum_i p_i E_{\bf f}(\ket{\psi_i}_{AB}),
\label{qEmixed}
\end{equation}
over all possible pure state decompositions of $\rho_{AB}=\sum_{i} p_i |\psi_i\rangle_{AB}\langle\psi_i|$.

For a probability ensemble $\mathcal E = \{p_i, \rho_i\}$ realizing a quantum state $\rho$ such that $\rho=\sum_{i}p_i\rho_i$, its {\em Holevo quantity} is defined as
\begin{align}
\chi\left(\mathcal E\right)=S\left(\rho\right)-\sum_{i}p_i S\left(\rho_i\right).
\label{eq: holevo}
\end{align}
Given a bipartite quantum state $\rho_{AB}$, each measurement $\{M^x_B\}$ applied on subsystem $B$ induces a probability ensemble $\mathcal E = \{p_x, \rho_A^x\}$ of the reduced density matrix $\rho_A=\T_A\rho_{AB}$ in the way that $p_x=\T[(I_A\otimes M_B^x)\rho_{AB}]$
is the probability of the outcome $x$ and $\rho^x_A=\T_B[(I_A\otimes {M_B^x})\rho_{AB}]/p_x$
is the state of system $A$ when the outcome was $x$.
The {\em one-way classical correlation}
(CC)~\cite{HV} of a bipartite state $\rho_{AB}$ is defined by the maximum Holevo quantity
\begin{align}
{\mathcal J}^{\leftarrow}(\rho_{AB})&= \max_{\mathcal E} \chi\left(\mathcal E\right)
\label{CC}
\end{align}
over all possible ensemble representations $\mathcal E$ of $\rho_A$ induced by measurements on subsystem $B$.

The following proposition shows a trade-off relation between classical correlation and quantum entanglement(measured by CC and EoF, respectively) distributed in three-party quantum systems.
\begin{Prop}~\cite{KW}
For a three-party pure state $\ket{\psi}_{ABC}$ with reduced density matrices $\rho_{AB}=\T_C\ket{\psi}_{ABC}\bra{\psi}$,
$\rho_{AC}=\T_B\ket{\psi}_{ABC}\bra{\psi}$ and $\rho_{A}=\T_{BC}\ket{\psi}_{ABC}\bra{\psi}$, we have
\begin{align}
S(\rho_A)={\mathcal J}^{\leftarrow}(\rho_{AB})+E_{\bf f}\left(\rho_{AC}\right).
\label{CCEq}
\end{align}
\label{Prop: KW}
\end{Prop}

\section{Classical-Classical-Quantum(CCQ) States}
\label{subsec: ccq}
In this section, we consider a four-party ccq states obtained from a bipartite state $\rho_{AB}$,
and provide detail evaluations of their mutual information. Without loss of generality, we assume that any bipartite state as a two-qu$d$it state by taking $d$ as the dimension of larger dimensional subsystem.

For a two-qudit state $\rho_{AB}$,  let us consider a spectral decomposition
\begin{align}
\rho_B=\sum_{i=0}^{d-1}\lambda_{i}\ket{e_i}_B\bra{e_i}
\label{specrhoB}
\end{align}
of the reduced density matrix $\rho_B=\T_{A}\rho_{AB}$. Let
\begin{align}
\mathcal E_0=\{\lambda_i,\sigma_A^i\}_i
\label{ensemble0}
\end{align}
be the probability ensemble of $\rho_A=\T_{B}\rho_{AB}$ from the measurement
$\{\ket{e_i}_B\bra{e_i}\}_{i=1}^{d-1}$ on subsystem $B$ of $\rho_{AB}$, in a way that
\begin{align}
\lambda_i=\T \left[(I_A \otimes\ket{e_i}_B\bra{e_i})\rho_{AB}\right]
\label{lami}
\end{align}
and
\begin{align}
\sigma_A^i=\frac{1}{\lambda_i}\T_B \left[(I_A \otimes\ket{e_i}_B\bra{e_i})\rho_{AB}\right].
\label{sigi}
\end{align}

Based on the eigenvectors $\{ \ket{e_j }_{B}\}$ of $\rho_B$, we also consider the $d$-dimensional {\em Fourier basis} elements
\begin{equation}
|\tilde e_j \rangle_{B} = \frac{1}{\sqrt{d}}\sum_{k=0}^{d-1}
\omega_d^{jk}\ket{e_k}_{B}
\label{fourier}
\end{equation}
for each $j=0,\ldots ,d-1$, where $\omega_d = e^{\frac{2\pi i}{d}}$ is the $d$th-root of unity.
Let
\begin{align}
\mathcal E_1=\{\frac{1}{d},\tau_A^j\}_j,
\label{ensemble1}
\end{align}
be the probability ensemble of $\rho_A=\T_{B}\rho_{AB}$ obtained by measuring subsystem $B$
in terms of the Fourier basis $\{\ket{\tilde{e}_j}_B\bra{\tilde{e}_j}\}_{j=1}^{d-1}$ where
\begin{align}
\frac{1}{d}=\T \left[(I_A \otimes\ket{\tilde{e}_j}_B\bra{\tilde{e}_j})\rho_{AB}\right]
\label{1/d}
\end{align}
and
\begin{align}
\tau_A^j=d\T_B\left[(I_A\otimes |\tilde e_j\rangle_B \langle\tilde e_j|)\rho_{AB}\right].
\label{tauj}
\end{align}

Now we define the generalized $d$-dimensional Pauli operators based on the eigenvectors of $\rho_B$ as
\begin{align}
Z=\sum_{j=0}^{d-1}\omega_d^j\ket{e_j}_{B}\bra{e_j},~X=\sum_{j=0}^{d-1}\ket{e_{j+1}}_{B}\bra{e_j}=\sum_{j=0}^{d-1} \omega_d^{-j}|\tilde
e_j \rangle_{B} \langle \tilde e_j |,
\label{paulis}
\end{align}
and consider a four-qudit ccq state $\Gamma_{XYAB}$
\begin{align}
\Gamma_{XYAB}=\frac {1}{d^2}\sum_{x,y=0}^{d-1}\ket{x}_X
\bra{x}\otimes\ket{y}_Y\bra{y}
\otimes \left[(I_A\otimes X^x_BZ^y_B)\rho_{AB}(I_A\otimes Z^{-y}_BX^{-x}_B)\right],
\label{XYAB}
\end{align}
for some $d$-dimensional orthonormal bases $\{\ket{x}_X\}$ and $\{\ket{y}_Y\}$ of the subsystems $X$ and $Y$, respectively.
From Eqs.~(\ref{lami}),~(\ref{sigi}),~(\ref{1/d}), (\ref{tauj}) and (\ref{paulis}), the reduced density matrices of $\Gamma_{XYAB}$ are obtained as
\begin{align}
\Gamma_{XAB}=\frac {1}{d}\sum_{x=0}^{d-1}&\ket{x}_X\bra{x}\otimes \left[ \left(I_A \otimes X^x_B\right)
\left(\sum_{i=0}^{d-1} \sigma_A^i \otimes \lambda_i\ket{e_i}_B\bra{e_i}\right)\left(I_A \otimes X_B^{-x}\right)\right],
\label{XAB}
\end{align}
\begin{align}
\Gamma_{YAB}=&\frac {1}{d}\sum_{y=0}^{d-1}\ket{y}_Y\bra{y}\otimes \left[ \left(I_A \otimes Z_B^y \right) \left(\sum_{j=0}^{d-1} \tau_A^j \otimes
\frac{1}{d}|\tilde e_j \rangle_B \langle \tilde e_j|\right)\left(I_A \otimes Z_B^{-y}\right)\right],
 \label{YAB}
\end{align}
and
\begin{equation}
\Gamma_{AB}=\rho_A\otimes\frac{I_B}{d},
\label{AB}
\end{equation}
where $I_A$ and $I_B$ are $d$-dimensional identity operators of subsystems $A$ and $B$, respectively.

Before we move to the next section, we evaluate the mutual information of the ccq state in Eq.~(\ref{XYAB}) as well as the reduced density matrices in
Eqs.~(\ref{XAB}) and (\ref{YAB}); the classical parts of the four-qudit ccq state $\Gamma_{XYAB}$ in Eq.~(\ref{XYAB}) is
\begin{align}
\Gamma_{XY}=\frac 1{d^2}\sum_{x,y=0}^{d-1}\ket{x}_X \bra{x}\otimes\ket{y}_Y\bra{y},
\label{GammaXY}
\end{align}
which is the maximally mixed state in $d^2$-dimensional quantum system, therefore its von Neumann entropy is
\begin{align}
S\left(\Gamma_{XY}\right)=-\sum_{x,y=0}^{d-1}\frac{1}{d^2}\ln\left(\frac{1}{d^2}\right)=2\ln d.
\label{vonOXY}
\end{align}
We also note that Eq.~(\ref{AB}) leads us to
\begin{align}
S\left(\Gamma_{AB} \right)=S(\rho_A)+\ln d.
\label{vonAB}
\end{align}
From the {\em joint entropy theorem}~\cite{joint1, joint},
we have
\begin{align}
S(\Gamma_{XYAB})=&2\ln d+\frac {1}{d^2}\sum_{x,y=0}^{d-1}S\left((I_A\otimes X^x_BZ^y_B)\rho_{AB}(I_A\otimes Z^{-y}_BX^{-x}_B)\right)
=2\ln d+S\left(\rho_{AB}\right),
\label{Ixyab0}
\end{align}
where the second equality is due to the unitary invariance of von Neumann entropy.
Thus Eqs.~(\ref{vonOXY}),~(\ref{vonAB}) and (\ref{Ixyab0}) give us
the mutual information of the four-qudit ccq state $\Gamma_{XYAB}$ with respect to the bipartition between $XY$ and $AB$ as
\begin{align}
{I}\left(\Gamma_{XY:AB}\right)=\ln d+S(\rho_A)-S(\rho_{AB}).
\label{Ixyab}
\end{align}

For the von Neumann entropy of $\Gamma_{XAB}$ in Eq.~(\ref{XAB}), we have
\begin{align}
S\left(\Gamma_{XAB}\right)=&\ln d +\frac {1}{d}\sum_{x=0}^{d-1} S\left(\left(I_A \otimes X^x_B\right)
\left(\sum_{i=0}^{d-1} \sigma_A^i \otimes \lambda_i\ket{e_i}_B\bra{e_i}\right)\left(I_A \otimes X_B^{-x}\right)\right)\nonumber\\
=&\ln d+S\left(\sum_{i=0}^{d-1} \sigma_A^i \otimes \lambda_i\ket{e_i}_B\bra{e_i}\right)\nonumber\\
=&\ln d+H(\Lambda)+\sum_{i=1}^{d-1}\lambda_i S\left(\sigma_A^i \right),
\label{vonXAB0}
\end{align}
where the first equality is from the joint entropy theorem, the second equality is due to the unitary invariance of von Neumann entropy
and the last equality is the joint entropy theorem together with
\begin{equation}
H(\Lambda)=-\sum_{i}\lambda_i \ln \lambda_i,
\label{shannonlamb}
\end{equation}
which is the shannon entropy of the spectrum $\Lambda=\{\lambda_i\}$ of $\rho_B$ in Eq.~(\ref{specrhoB}).
Thus we can rewrite the von Neumann entropy of $\Gamma_{XAB}$ as
\begin{align}
S\left(\Gamma_{XAB}\right)=\ln d+S\left(\rho_B\right)+\sum_{i=1}^{d-1}\lambda_i S\left(\sigma_A^i \right).
\label{vonXAB}
\end{align}
Because the classical parts of $\Gamma_{XAB}$ is the $d$-dimensional maximally mixed state $\Gamma_X=\frac{1}{d}\sum_{x=0}^{d-1}\ket{x}_X \bra{x}$,
we have the mutual information of $\Gamma_{XAB}$ with respect to the bipartition between $X$ and $AB$ as
\begin{align}
{I}(\Gamma_{X:AB})=&S(\Gamma_X)+S(\Gamma_{AB})-S(\Gamma_{XAB})
=\ln d-S(\rho_B)+\chi(\mathcal E_0).
\label{Ixab}
\end{align}

For the von Neumann entropy of $\Gamma_{YAB}$ in Eq.~(\ref{YAB}), we have
\begin{align}
S\left(\Gamma_{XAB}\right)=&\ln d +\frac {1}{d}\sum_{y=0}^{d-1} S\left(\left(I_A \otimes Z_B^y \right) \left(\sum_{j=0}^{d-1} \tau_A^j \otimes \frac{1}{d}|\tilde e_j \rangle_B \langle \tilde e_j|\right)\left(I_A \otimes Z_B^{-y}\right)\right)\nonumber\\
=&\ln d+S\left(\sum_{j=0}^{d-1} \tau_A^j \otimes \frac{1}{d}|\tilde e_j \rangle_B \langle \tilde e_j|\right)\nonumber\\
=&2\ln d+\frac{1}{d}\sum_{j=1}^{d-1} S\left(\tau_A^j \right),
\label{vonYAB0}
\end{align}
where the first and third equalities are due to the joint entropy theorem and the second equality is from the unitary invariance of von Neumann entropy. Thus the mutual information of $\Gamma_{YAB}$ with respect to the bipartition between $Y$ and $AB$ is
\begin{align}
{I}(\Gamma_{Y:AB})=&S(\Gamma_Y)+S(\Gamma_{AB})-S(\Gamma_{YAB})\nonumber\\
=&\ln d + S(\rho_A)+\ln d-2\ln d-\frac{1}{d}\sum_{j=1}^{d-1} S\left(\tau_A^j \right)\nonumber\\
=&\chi(\mathcal E_1).
\label{Iyab}
\end{align}

\section{Monogamy inequality of multi-party entanglement in terms of EoF}
\label{Sec: mono EoF}
It is known that quantum mutual information is superadditive for any ccq state~\cite{Kim16T}
\begin{equation}
  \Xi_{XYAB}=\frac 1{d^2}\sum_{x,y=0}^{d-1}\ket{x}_X
  \bra{x}\otimes\ket{y}_Y\bra{y}\otimes\sigma^{xy}_{AB},
\label{gXYAB}
\end{equation}
that is,
\begin{equation}
{I}\left(\Xi_{XY:AB}\right)\geq {I}\left(\Xi_{X:AB}\right)+{I}\left(\Xi_{Y:AB}\right).
\label{mysuper}
\end{equation}
Here we show that the additivity of quantum mutual information for ccq states guarantees the monogamy inequality of three-party quantum entanglement in therms of EoF.
\begin{Thm}
For any three-party pure state $\ket{\psi}_{ABC}$
with its two-qudit reduced density matrices $\T_C \ket{\psi}_{ABC}\bra{\psi}=\rho_{AB}$ and $\T_B \ket{\psi}_{ABC}\bra{\psi}=\rho_{AC}$,
we have
\begin{align}
E_{\bf f}\left(\ket{\psi}_{A(BC)}\right)
\geq& E_{\bf f}\left(\rho_{AB}\right)+E_{\bf f}\left(\rho_{AC}\right),
\label{3mono}
\end{align}
conditioned on the additivity of quantum mutual information for the ccq states in Eq.~(\ref{XYAB})
obtained by $\rho_{AB}$ and $\rho_{AC}$.
\label{thm: 3mono}
\end{Thm}

\begin{proof}
Let us first consider the four-qudit ccq state $\Gamma_{XYAB}$ of the form in Eq.~(\ref{XYAB}) obtained by the two-qudit
reduced density matrix $\rho_{AB}$ of $\ket{\psi}_{ABC}$. From Eqs.~(\ref{Ixyab}),~(\ref{Ixab}) and (\ref{Iyab}),
the additivity condition of quantum mutual information for $\Gamma_{XYAB}$
\begin{align}
{I}\left(\Gamma_{XY:AB}\right)= {I}\left(\Gamma_{X:AB}\right)+{I}\left(\Gamma_{Y:AB}\right)
\label{maddab}
\end{align}
can be rewritten as
\begin{align}
\chi(\mathcal E_0)+\chi(\mathcal E_1)=S\left(\rho_A\right)+S\left(\rho_B\right)-S\left(\rho_{AB}\right)
={I}\left(\rho_{AB}\right),
\label{maddab2}
\end{align}
where $\mathcal E_0$ and $\mathcal E_1$ are the probability ensembles of
$\rho_{A}$ in Eqs.~(\ref{ensemble0}) and (\ref{ensemble1}), respectively.

Because $\mathcal E_0$ and $\mathcal E_1$ can be obtained from measuring subsystem $B$ of $\rho_{AB}$ by the rank-1 measurement $\{\ket{e_i}_B\bra{e_i}\}_{i=1}^{d-1}$ and $\{\ket{\tilde e_j}_B\bra{\tilde e_j}\}_{j=1}^{d-1}$, respectively, the definition of CC in Eq.~(\ref{CC}) leads us to
\begin{align}
{\mathcal J}^{\leftarrow}(\rho_{AB})\geq \chi(\mathcal E_0),~{\mathcal J}^{\leftarrow}(\rho_{AB})\geq \chi(\mathcal E_1),
\label{CClow0}
\end{align}
therefore
\begin{align}
{\mathcal J}^{\leftarrow}(\rho_{AB})\geq \frac{1}{2}\left(\chi(\mathcal E_0)+\chi(\mathcal E_1)\right)
=\frac{1}{2} {I}\left(\rho_{AB}\right).
\label{CClowall1}
\end{align}
By considering the ccq state $\Gamma_{XYAC}$ from the two-qudit reduced density matrix
$\rho_{AC}$ of $\ket{\psi}_{ABC}$, we can analogously have
\begin{align}
{\mathcal J}^{\leftarrow}(\rho_{AC})\geq&\frac{1}{2} {I}\left(\rho_{AC}\right).
\label{CClowall2}
\end{align}

As the trade-off relation of Eq.~(\ref{CCEq}) in Proposition~\ref{Prop: KW} is universal with respect to
the subsystems, we also have
\begin{align}
S(\rho_A)={\mathcal J}^{\leftarrow}(\rho_{AC})+E_{\bf f}\left(\rho_{AB}\right),
\label{CCEq2}
\end{align}
for the given two-qudit state $\ket{\psi}_{ABC}$, therefore
\begin{align}
E_{\bf f}\left(\rho_{AB}\right)+E_{\bf f}\left(\rho_{AC}\right)=&
2S(\rho_A)-\left({\mathcal J}^{\leftarrow}(\rho_{AB})+{\mathcal J}^{\leftarrow}(\rho_{AC})\right).
\label{uni1}
\end{align}
Now Inequalities~(\ref{CClowall1}), (\ref{CClowall2}) as well as Eq.~(\ref{uni1}) lead us to
\begin{align}
E_{\bf f}\left(\rho_{AB}\right)+E_{\bf f}\left(\rho_{AC}\right)\leq&
2S(\rho_A)-\frac{1}{2}\left({I}(\rho_{AB})+{I}(\rho_{AC})\right)\nonumber\\
=&2S(\rho_A)-\frac{1}{2}\left(S(\rho_A)+S(\rho_B)-S(\rho_{AB})+S(\rho_A)+S(\rho_C)-S(\rho_{AC})\right)\nonumber\\
=&S(\rho_A)\nonumber\\
=&E_{\bf f}\left(\ket{\psi}_{A(BC)}\right),
\label{uni2}
\end{align}
where the second equality is due to $\rho_{AC}=\rho_B$ and $\rho_{AB}=\rho_C$ for three-party pure state $\ket{\psi}_{ABC}$.
\end{proof}

To illustrate Theorem~\ref{thm: 3mono}, let us first consider three-qubit GHZ state~\cite{GHZ},
\begin{align}
\ket{GHZ}_{ABC}=\frac{1}{\sqrt 2}\left(\ket{000}_{ABC}+\ket{111}_{ABC}\right),
\label{3GHZ}
\end{align}
with its reduced density matrices
\begin{align}
\rho_{AB}=\frac{1}{2}\left(\ket{00}_{AB}\bra{00}+\ket{11}_{AB}\bra{11}\right),~\rho_{A}=\frac{1}{2}\left(\ket{0}_{A}\bra{0}+\ket{1}_{A}\bra{1}\right)
\label{GHZredaba}
\end{align}
and
\begin{align}
\rho_{B}=\frac{1}{2}\left(\ket{0}_{B}\bra{0}+\ket{1}_{B}\bra{1}\right).
\label{GHZredb}
\end{align}

The eigenvalues of $\rho_B$ are $\lambda_0=\lambda_1=\frac{1}{2}$ with corresponding eigenvectors $\ket{e_0}_B=\ket{0}_B$ and $\ket{e_1}_B=\ket{1}_B$ respectively. Thus the ensemble of $\rho_A$ induced by measuring subsystem $B$ of $\rho_{AB}$ in terms of the eigenvectors of $\rho_B$, that is, $\{ M_B^0=\ket{0}_B\bra{0}, M_B^1=\ket{1}_B\bra{1}\}$ is
\begin{align}
\mathcal E_0= \{ \lambda_0=\frac{1}{2}, \sigma_A^0=\ket{0}_A \bra{0}, \lambda_1=\frac{1}{2}, \sigma_A^1=\ket{1}_A\bra{1}\}.
\label{GHZE0}
\end{align}
Because the Fourier basis elements of subsystem $B$ with respect to the eigenvectors of $\rho_B$ are
\begin{align}
|\tilde e_0 \rangle_B=\frac{1}{2}\left(\ket{0}_B+\ket{1}_B\right),
~|\tilde e_1 \rangle_B=
\frac{1}{2}\left(\ket{0}_B-\ket{1}_B\right), \label{Four2}
\end{align}
the ensemble of $\rho_A$ induced by measuring subsystem $B$ of $\rho_{AB}$ in terms of the Fourier basis in Eq.~(\ref{Four2}) is
\begin{align}
\mathcal E_1= \{ \frac{1}{2}, \tau_A^0=\frac{1}{2}I_A, \frac{1}{2}, \tau_A^1=\frac{1}{2}I_A\}.
\label{GHZE1}
\end{align}

Now we consider the additivity of mutual information of the ccq state $\Gamma_{XYAB}$ obtained from $\rho_{AB}$ in Eq.~(\ref{GHZredaba}). Due to Eq.~(\ref{Ixyab}), the mutual information of $\Gamma_{XYAB}$ between $XY$ and $AB$ is
\begin{align}
I\left(\Gamma_{XY:AB}\right)=&\ln 2+S(\rho_A)-S(\rho_{AB})=\ln 2
\label{GHZmutxyab}
\end{align}
because $S(\rho_A)=S(\rho_{AB})=\ln 2$ from Eqs.~(\ref{GHZredaba}).
For the mutual information of $\Gamma_{XAB}$ between $X$ and $AB$, Eq.~(\ref{Ixab}) leads us to
\begin{align}
I\left(\Gamma_{X:AB}\right)=&\ln 2-S(\rho_B)+\chi(\mathcal E_0)=\ln 2
\label{GHZmutxab}
\end{align}
where the second equality is from $S(\rho_B)= \ln 2$ and
\begin{align}
\chi(\mathcal E_0)=S(\rho_A)-\frac{1}{2}S(\ket{0}_A \bra{0})-\frac{1}{2}S(\ket{1}_A \bra{1})=\ln 2,
\label{ghzholev0}
\end{align}
for the ensemble $\mathcal E_0$ in Eq.~(\ref{GHZE0}).
For the mutual information of $\Gamma_{YAB}$ between $Y$ and $AB$, Eq.~(\ref{Iyab}) leads us to
\begin{align}
I\left(\Gamma_{Y:AB}\right)=&\chi(\mathcal E_1)
=S(\rho_A)-\frac{1}{2}S\left(\frac{1}{2}I_A \right)-\frac{1}{2}S\left(\frac{1}{2}I_A \right)
=0
\label{GHZmutyab}
\end{align}
where the second equality is due to the ensemble $\mathcal E_1$ in Eq.~(\ref{GHZE1}).

From Eqs.~(\ref{GHZmutxyab}), (\ref{GHZmutxab}) and (\ref{GHZmutyab}), we note that the mutual information of the  ccq state $\Gamma_{XYAB}$ obtained from $\rho_{AB}$ is additive as in Eq.~(\ref{maddab}). Moreover, the symmetry of GHZ state assures that the same is also true for the reduced density matrix $ \rho_{AC}=\T_B \ket{GHZ}_{ABC}\bra{GHZ}$. Thus Theorem~\ref{thm: 3mono} guarantees the monogamy inequality of the three-qubit GHZ state in Eq.~(\ref{3GHZ}) in terms of EoF. In fact, we have
\begin{align}
E_{\bf f}\left(\ket{GHZ}_{A(BC)}\right)=S(\rho_A)=\ln 2
\label{GHZEfABC}
\end{align}
whereas the two-qubit reduced density matrices $\rho_{AB}$ and $\rho_{AC}$ are separeble. Thus
\begin{align}
E_{\bf f}\left(\rho_{AB}\right)=E_{\bf f}\left(\rho_{AC}\right)=0,
\end{align}
and this implies the monogamy inequality in (\ref{3mono}).

Let us consider another example of three-qubit state; three-qubit W-state is defined as~\cite{W}
\begin{align}
\ket{W}_{ABC}=\frac{1}{\sqrt 3}\left(\ket{100}_{ABC}+\ket{010}_{ABC}+\ket{001}_{ABC}\right).
\label{3W}
\end{align}
The two-qubit reduced density matrix of $\ket{W}_{ABC}$ on subsystem $AB$ is obtained as
\begin{align}
\rho_{AB}=\frac{2}{3}\ket{\psi^+}_{AB}\bra{\psi^+}+\frac{1}{3}\ket{00}_{AB}\bra{00}
\label{Wredab}
\end{align}
where
\begin{align}
\ket{\psi^+}_{AB}=\frac{1}{\sqrt 2}\left(\ket{01}_{AB}+\ket{10}_{AB}\right)
\label{psi+}
\end{align}
is the two-qubit Bell state, and the one-qubit reduced density matrices of $\rho_{AB}$ are
\begin{align}
\rho_{A}=\frac{2}{3}\ket{0}_{A}\bra{0}+\frac{1}{3}\ket{1}_{A}\bra{1},~
\rho_{B}=\frac{2}{3}\ket{0}_{B}\bra{0}+\frac{1}{3}\ket{1}_{B}\bra{1}.
\label{Wredb}
\end{align}

From to the spectral decomposition of $\rho_B$ in Eq.~(\ref{Wredb}) with the eigenvalues $\lambda_0=\frac{2}{3}, \lambda_1=\frac{1}{3}$ and corresponding eigenvectors $\ket{e_0}_B=\ket{0}_B$ and $\ket{e_1}_B=\ket{1}_B$, respectively, it is straightforward to check that the ensemble of $\rho_A$ induced from measuring subsystem $B$ of $\rho_{AB}$ by the eigenvectors of $\rho_B$ is
\begin{align}
\mathcal E_0= \{ \lambda_0=\frac{2}{3}, \sigma_A^0=\frac{1}{2}I_A, \lambda_1=\frac{1}{3}, \sigma_A^1=\ket{0}_A\bra{0}\}.
\label{WE0}
\end{align}
Because the Fourier basis of subsystem $A$ is the same as Eq.~(\ref{Four2}), it is also straightforward to obtain the ensemble of $\rho_A$ induced by measuring subsystem $B$ of $\rho_{AB}$ in terms of the Fourier basis,
\begin{align}
\mathcal E_1= \{ \frac{1}{2}, \tau_A^j\}_{j=1,2},
\label{WE1}
\end{align}
where
\begin{align}
\tau_A^0=&\frac{1}{3}\left(2\ket{0}_{A}\bra{0}+\ket{0}_{A}\bra{1}+\ket{1}_{A}\bra{0}+\ket{1}_{A}\bra{1}\right),~
\tau_A^1=\frac{1}{3}\left(2\ket{0}_{A}\bra{0}-\ket{0}_{A}\bra{1}-\ket{1}_{A}\bra{0}+\ket{1}_{A}\bra{1}\right).
\label{wtau}
\end{align}

For the mutual information of the ccq state $\Gamma_{XYAB}$ obtained from $\rho_{AB}$ in Eq.~(\ref{Wredab}), Eq.~(\ref{Ixyab}) together with Eqs.~(\ref{Wredab}) and (\ref{Wredb}) lead us to
\begin{align}
I\left(\Gamma_{XY:AB}\right)=&\ln 2+S(\rho_A)-S(\rho_{AB})=\ln 2.
\label{Wmutxyab}
\end{align}
For the mutual information of $\Gamma_{XAB}$ between $X$ and $AB$, Eq.~(\ref{Ixab}) leads us to
\begin{align}
I\left(\Gamma_{X:AB}\right)=&\ln 2-S(\rho_B)+\chi(\mathcal E_0),
\label{Wmutxab0}
\end{align}
where Eq.~(\ref{WE0}) impies
\begin{align}
\chi(\mathcal E_0)=S(\rho_A)-\frac{2}{3}S\left( \frac{1}{2}I_A \right)-\frac{1}{3}S\left( \ket{0}_A\bra{0}\right)=S(\rho_A)-\frac{2}{3}\ln 2.
\label{chie0}
\end{align}
Due to Eq.~(\ref{Wredb}), we have $S(\rho_A)=S(\rho_B)$, therefore Eqs.~(\ref{Wmutxab}) and (\ref{chie0}) lead us to
\begin{align}
I\left(\Gamma_{X:AB}\right)=\frac{1}{3}\ln 2.
\label{Wmutxab}
\end{align}

For the mutual information of $\Gamma_{YAB}$ between $Y$ and $AB$, Eq.~(\ref{Iyab}) leads us to
\begin{align}
I\left(\Gamma_{Y:AB}\right)=&\chi(\mathcal E_1)=S(\rho_A)-\frac{1}{2}S\left(\tau_A^0 \right)-\frac{1}{2}S\left(\tau_A^1 \right),
\label{Wmutyab0}
\end{align}
where the second equality is from the ensemble $\mathcal E_1$ in Eq.~(\ref{WE1}).
Here we note that $\tau_A^0$ and $\tau_A^1$ in Eq.~(\ref{wtau}) have the same eigenvalues, that is
$\mu_0=\frac{3+\sqrt 5}{6}$ and $\mu_0=\frac{3-\sqrt 5}{6}$,
therefore we have $S\left(\tau_A^0 \right)=S\left(\tau_A^1 \right)$.
From the spectral decomposition of $\rho_A$ in Eq.~(\ref{Wredb}), we have
\begin{align}
S\left(\rho_A\right)=\ln 3 -\frac{2}{3}\ln 2
\label{SrhoA}
\end{align}
and this turns Eq.~(\ref{Wmutyab0}) into
\begin{align}
I\left(\Gamma_{Y:AB}\right)=\ln 3 -\frac{2}{3}\ln 2-S\left(\tau_A^0 \right).
\label{Wmutyab}
\end{align}

From Eqs.~(\ref{Wmutxyab}), (\ref{Wmutxab}) and (\ref{Wmutyab}), we have
\begin{align}
I\left(\Gamma_{XY:AB}\right)-I\left(\Gamma_{X:AB}\right)-I\left(\Gamma_{Y:AB}\right)=\frac{4}{3}\ln 2 -\ln 3 +S\left(\tau_A^0 \right),
\label{nonadd}
\end{align}
where
\begin{align}
\ln 2 \approx 0.693147,~ \ln 3 \approx 1.098612,
\label{approxln}
\end{align}
and
\begin{align}
S\left(\tau_A^0 \right)=-\mu_0 \ln \mu_0-\mu_1 \ln \mu_1 \approx 0.381264.
\label{approxSt}
\end{align}
Thus we have
\begin{align}
I\left(\Gamma_{XY:AB}\right)-I\left(\Gamma_{X:AB}\right)-I\left(\Gamma_{Y:AB}\right)\approx 0.206848 >0,
\label{nonadd1}
\end{align}
which implies the nonadditivity of mutual information for the ccq state $\Gamma_{XYAB}$ obtained from $\rho_{AB}$ in Eq.~(\ref{Wredab}).
We also note that the symmetry of W state in Eq.~(\ref{3W}) would imply the nonadditivity of mutual information for the ccq state $\Gamma_{XYAC}$ obtained from the two-qubit reduced density matrix $\rho_{AC}$ of W state in Eq.~(\ref{3W}).

As the additivity of mutual information in Theorem~\ref{thm: 3mono} is only a sufficient condition for monogamy inequality in terms of EoF,
nonadditivity does not directly imply violation of Inequality~(\ref{3mono}) for the W state in Eq.~(\ref{3W}). However, we note that $\rho_{AB}$
in Eq.~(\ref{Wredab}) is a two-qubit state, therefore its EoF can be analytically evaluated as~\cite{wootters}
\begin{align}
E_{\bf f}\left(\rho_{AB}\right)\approx 0.3812.
\label{EoFWAB}
\end{align}
Moreover, the symmetry of the W state assures that the EoF of $\rho_{AC}=\T_B \ket{W}_{ABC}\bra{W}$ is the same,
\begin{align}
E_{\bf f}\left(\rho_{AC}\right)\approx 0.3812,
\label{EoFWAC}
\end{align}
whereas
\begin{align}
E_{\bf f}\left(\ket{W}_{A(BC)}\right)=S(\rho_A)\approx0.6365.
\label{EoFWABC}
\end{align}
As Eqs.~(\ref{EoFWAB}), (\ref{EoFWAC}) and (\ref{EoFWABC}) imply the violation of Inequality~(\ref{3mono}), W state in Eq.~(\ref{3W}) can be considered as an example for the contraposition of Theorem~\ref{thm: 3mono}; violation of monogamy inequality in (\ref{3mono}) implies nonadditivity of quantum mutual information for the ccq state.

Now, we generalize Theorem~\ref{thm: 3mono} for multi-party quantum states of arbitrary dimension.
\begin{Thm}
For any multi-party quantum state $\rho_{A_1A_2\cdots A_n}$ with two-party reduced density matrices
$\rho_{A_1A_i}$ for $i=2, \cdots , n$, we have
\begin{align}
E_{\bf f}\left(\rho_{A_1(A_2\cdots A_n)}\right)
\geq& \sum_{i=2}^{n}E_{\bf f}\left(\rho_{A_1A_i}\right),
\label{monon}
\end{align}
conditioned on the additivity of quantum mutual information for the ccq states in the form of Eq.~(\ref{XYAB})
obtained by $\rho_{A_1A_i}$ for $i=2, \cdots , n$.
\label{thm: monon}
\end{Thm}

\begin{proof}
We first prove the theorem for any three-party mixed state $\rho_{ABC}$, and inductively show the validity of the theorem for any $n$-party quantum state $\rho_{A_1A_2\cdots A_n}$.
For a three-party mixed state $\rho_{ABC}$, let us consider an optimal decomposition of $\rho_{ABC}$ realizing EoF with respect to the bipartition between $A$ and $BC$, that is,
\begin{align}
\rho_{ABC}=\sum_i p_i\ket{\psi_i}_{ABC}\bra{\psi_i},
\label{opt1}
\end{align}
with
\begin{align}
E_{\bf f}\left(\rho_{A(BC)}\right)=\sum_i p_i E_{\bf f}\left(\ket{\psi_i}_{A(BC)}\right).
\label{optEoF1}
\end{align}

From Theorem~\ref{thm: 3mono}, each pure state $\ket{\psi_i}_{ABC}$ of the decomposition (\ref{opt1}) satisfies
\begin{align}
E_{\bf f}\left(\ket{\psi_i}_{A(BC)}\right)\geq E_{\bf f}\left(\rho^i_{AB}\right)+E_{\bf f}\left(\rho^i_{AC}\right)
\label{polyi}
\end{align}
with $\rho^i_{AB}=\T_C \ket{\psi_i}_{ABC}\bra{\psi_i}$ and $\rho^i_{AC}=\T_B \ket{\psi_i}_{ABC}\bra{\psi_i}$,
therefore,
\begin{align}
E_{\bf f}\left(\rho_{A(BC)}\right)=&\sum_i p_i E_{\bf f}\left(\ket{\psi_i}_{A(BC)}\right)
\geq \sum_i p_i E_{\bf f}\left(\rho^i_{AB}\right)+ \sum_i p_i E_{\bf f}\left(\rho^i_{AC}\right).
\label{eqmono1}
\end{align}

For each $i$ and the two-party reduced density matrices $\rho_{AB}^i$, let us consider their optimal decompositions for EoF, that is,
\begin{align}
\rho_{AB}^i=&\sum_{j}r_{ij}\ket{\mu_j^i}_{AB}\bra{\mu_j^i}~,
\label{optrhoAB^i}
\end{align}
such that
\begin{align}
E_{\bf f}\left(\rho_{AB}^i\right)=&\sum_{j}r_{ij} E_{\bf f}\left(\ket{\mu_j^i}_{AB}\right).
\label{ErhoAB^i}
\end{align}
Now we have
\begin{align}
\sum_i p_i E_{\bf f}\left(\rho^i_{AB}\right)
=\sum_{i,j} p_i r_{ij} E_{\bf f}\left(\ket{\mu_j^i}_{AB}\right)
\geq E_{\bf f}\left(\rho_{AB}\right),
\label{ot2AB}
\end{align}
where the inequality is due to
\begin{align}
\rho_{AB}=\sum_i p_i \rho_{AB}^i=&\sum_{i,j}p_i r_{ij}\ket{\mu_j^i}_{AB}\bra{\mu_j^i}
\label{decompallAB}
\end{align}
and the definition of EoF.

For each $i$, we also consider an optimal decomposition of $\rho_{AC}$
\begin{align}
\rho_{AC}^i=&\sum_{l}s_{il}\ket{\nu_l^i}_{AC}\bra{\nu_l^i},
\label{optrhoAC^i}
\end{align}
for each $i$, such that
\begin{align}
E_{\bf f}\left(\rho_{AC}^i \right)=&\sum_{l}s_{il} E_{\bf f}\left(\ket{\nu_l^i}_{AC}\right).
\label{ErhoAC^i}
\end{align}
We can analogously have
\begin{align}
\sum_i p_i E_{\bf f}\left(\rho^i_{AC}\right)
\geq& E_{\bf f}\left(\rho_{AC}\right),
\label{ot2AC}
\end{align}
due to
\begin{align}
\rho_{AC}=\sum_i p_i \rho_{AC}^i=&\sum_{i,l}p_i s_{il}\ket{\nu_l^i}_{AC}\bra{\nu_l^i},
\label{decompallAC}
\end{align}
and the definition of EoF.
From Inequalities~(\ref{eqmono1}), (\ref{ot2AB}) and (\ref{ot2AC}), we have
\begin{align}
E_{\bf f}\left(\rho_{A(BC)}\right)\geq E_{\bf f}\left(\rho_{AB}\right)+E_{\bf f}\left(\rho_{AC}\right),
\label{3monomix}
\end{align}
which proves the theorem for three-party mixed states.

For general multi-party quantum system, we use the mathematical induction on the number of parties $n$; let us assume Inequality~(\ref{monon}) is true for any $k$-party quantum state, and consider an $k+1$-party quantum state $\rho_{A_1A_2\cdots A_{k+1}}$ for $k \geq 3$. By considering $\rho_{A_1A_2\cdots A_{k+1}}$ as a three-party state with respect to the tripartition $A_1$, $A_2\cdots A_k$ and $A_{k+1}$,
Inequality~(\ref{3monomix}) leads us to
\begin{align}
E_{\bf f}\left(\rho_{A_1(A_2\cdots A_{k+1})}\right)
\geq& E_{\bf f}\left(\rho_{A_1(A_2\cdots A_k)}\right)+E_{\bf f}\left(\rho_{A_1A_{k+1}}\right).
\label{kmonomixed1}
\end{align}

As $\rho_{A_1A_2\cdots A_k}$ in Inequality~(\ref{kmonomixed1}) is a
$k$-party quantum state, the induction hypothesis assures that
\begin{align}
E_{\bf f}\left(\rho_{A_1(A_2\cdots A_k)}\right) \geq
E_{\bf f}\left(\rho_{A_1A_2}\right)+\cdots +E_{\bf f}\left(\rho_{A_1A_k}\right).
\label{kmonomixed2}
\end{align}
Now Inequalities~(\ref{kmonomixed1}) and (\ref{kmonomixed2}) lead us to the monogamy
inequality
\begin{align}
E_{\bf f}\left(\rho_{A_1(A_2\cdots A_{k+1})}\right)
\geq& E_{\bf f}\left(\rho_{A_1A_2}\right)+\cdots +E_{\bf f}\left(\rho_{A_1A_{k+1}}\right),
\label{kmonomixed3}
\end{align}
which completes the proof.
\end{proof}


\section{Discussion}\label{Sec: Discussion}
We have considered possible conditions for monogamy inequality of multi-party quantum entanglement in terms of EoF, and shown that the additivity of mutual information of the ccq states implies the monogamy inequality of three-party quantum entanglement in terms of EoF. We have also provided examples of three-qubit GHZ and W states to illustrate our result in three-party case, and generalized our result into any multi-party systems of arbitrary dimensions.

Most monogamy inequalities of quantum entanglement deal with bipartite entanglement measures based on the minimization over all possible pure state ensembles. As analytic evaluation of such entanglement measure is generally hard especially in higher dimensional quantum systems more than qubits, the situation becomes far more difficult in investigating and establishing entanglement monogamy of multi-party quantum systems of arbitrary dimensions. The sufficient condition provided here deals with the quantum mutual information of the ccq states to guarantee the monogamy inequality of entanglement in terms of EoF in arbitrary dimensions. As the sufficient condition is not involved with any minimization process, our result can provide a useful methodology to understand the monogamy nature of multi-party quantum entanglement in arbitrary dimensions. We finally remark that it would be an interesting future task to investigate if the condition provided here is also necessary.

\section*{Acknowledgments}
This work was supported by Basic Science Research Program(NRF-2020R1F1A1A010501270) and Quantum Computing Technology Development Program(NRF-2020M3E4A1080088) through the National Research Foundation of Korea(NRF) grant funded by the Korea government(Ministry of Science and ICT).



\begin{thebibliography}{1}
\bibitem{joint}
For any orthonormal basis $\{\ket{i}_A\}$ and probability ensemble $\{p_i, \rho_B^i\}$,
\begin{equation}
S\left(\sum_{i}p_i \ket{i}_A\bra{i} \otimes \rho_B^i\right)=H\left(\mathbf{P}\right)+\sum_{i}p_{i}S\left(\rho_B^i\right),\nonumber
\label{eq: joint}
\end{equation}
where
\begin{equation}
H\left(\mathbf{P}\right)=-\sum_{i}p_i \ln p_i\nonumber
\label{shannonent}
\end{equation}
is the shannon entropy of the probability distribution $\mathbf{P}=\{p_i\}$.
\end{thebibliography}
\end{document}